\documentclass[11pt,oneside]{amsart}
\usepackage{latexsym}
\usepackage{epsfig}
\usepackage{amsmath}
\usepackage{amsfonts}
\usepackage{exscale}
\usepackage{graphicx}

\newtheorem{theorem}{Theorem}[section]
\newtheorem{cor}[theorem]{Corollary}
\newtheorem{prop}[theorem]{Proposition}

\newtheorem{remark}[theorem]{Remark}
\newtheorem{lemma}[theorem]{Lemma}
\newtheorem{defn}{Definition}

\newcommand{\sm}{\left(\begin{smallmatrix}}
\newcommand{\esm}{\end{smallmatrix}\right)}
\newcommand{\eps}{\varepsilon}
\newcommand{\divg}{\operatorname {div}}

\newcommand{\Ran}{\operatorname{ran}}

\newcommand{\ck}{\mathcal{K}}

\renewcommand{\eps}{\varepsilon}

\newcommand{\R}{\ensuremath{\mathbb{R}}}
\newcommand{\N}{\ensuremath{\mathbb{N}}}
\newcommand{\Z}{\ensuremath{\mathbb{Z}}}

\newcommand{\lo}{\mathfrak{L_0}}
\newcommand{\go}{\mathfrak{G_0}}
\newcommand{\lp}{\mathfrak{L_1}}
\newcommand{\gp}{\mathfrak{G_1}}

\newcommand{\norm}[1]{\left\Vert#1\right\Vert}

\newcommand{\la}{\left\langle}
\newcommand{\ra}{\right\rangle}
\newcommand{\llangle}{\left\langle}
\newcommand{\rrangle}{\right\rangle}

\newcommand{\bx}{\mathbf{x}}
\newcommand{\bX}{\mathbf{X}}
\newcommand{\bz}{\mathbf{z}}
\newcommand{\by}{\mathbf{y}}

\newcommand{\bk}{\mathbf{k}}

\newcommand{\be}{\begin{equation*}}
\newcommand{\ee}{\end{equation*}}
\newcommand{\bea}{\begin{eqnarray*}}
\newcommand{\eea}{\end{eqnarray*}}
\newcommand{\ben}{\begin{eqnarray}}
\newcommand{\een}{\end{eqnarray}}
\newcommand{\beq}{\begin{equation}}
\newcommand{\eeq}{\end{equation}}
\newcommand{\enq}{\end{equation}}

\title[Gap localization of TE-Modes]{Gap localization of TE-Modes by arbitrarily weak defects}

\author{B.M. Brown}
\address{B.M. Brown, Cardiff School of Computer Science, Cardiff University, Cardiff, CF24 3AA, Wales, UK}
\email{Malcolm.Brown@cs.cardiff.ac.uk}
\author{V. Hoang}
\address{V. Hoang, Rice University, Department of Mathematics-MS 136, Box 1892, Houston, TX 77251-1892}
\email{Vu.Hoang@rice.edu}
\author{M. Plum}
\address{M. Plum, Institute for Analysis,
Karlsruhe Institute of Technology (KIT), Kaiserstrasse 89,
Karlsruhe, Germany}
\email{michael.plum@kit.edu}
\author{M. Radosz}
\address{M. Radosz, Rice University, Department of Mathematics-MS 136, Box 1892, Houston, TX 77251-1892}
\email{maria\_radosz@hotmail.com}
\author{I. Wood}
\address{I. Wood, School of Mathematics, Statistics and Actuarial Sciences,
University of Kent, Canterbury, CT2 7NF, UK}
\email{i.wood@kent.ac.uk}
\date{\today}

\begin{document}
\maketitle

\begin{abstract}
This paper considers the propagation of TE-modes in photonic crystal waveguides. The waveguide is created by introducing a linear defect into a periodic background medium. Both the periodic background problem and the perturbed problem are modelled by a divergence type equation. A feature of our analysis is that we allow discontinuities in the coefficients of the operator, which is required to model many photonic crystals. It is shown that arbitrarily weak perturbations introduce spectrum into the spectral gaps of the background operator.
\end{abstract}

%
\section{Introduction}
%

Electromagnetic waves in periodically structured media, such as photonic crystals
and metamaterials, are a subject of ongoing interest. Typically, the propagation of
waves in such media exhibit \emph{band-gaps} (see \cite{Joann,KuchRev}), i.e. intervals on the frequency
or energy axis where propagation is forbidden. Mathematically, these correspond to 
gaps in the spectrum of the operator describing a problem with periodic background medium.
The existence of these gaps for certain choices of material coefficients
was proved in \cite{CD99,FK96,HPW09} and in \cite{Fil03} for the full
Maxwell case. Using layer potential techniques this question has been studied in \cite{AKSZ06,AKL09,AKL09book}.

In this paper, we consider the propagation of TE-polarized waves in 
photonic crystals. TE-polarization (transverse electric)
here means that the direction of the electric field is confined to a plane perpendicular to the
direction
of the magnetic field. When the periodicity is perturbed by point or line defects,
localization may take place in band gaps, analogous to the situation in
solid-state physics and semiconductor devices. The use of line defects in photonic
crystals has been proposed in the context of wave guide applications. The gap
localization gives rise to \emph{guided modes} which decay exponentially into
the bulk structure and propagate along the direction of the line defect.
For this reason, it is of great importance to know whether a given line defect
produces gap modes.

Is it possible to give rigorous sufficient conditions which imply localization
in gaps? In particular, does localization also occur when arbitrarily weak
defects are introduced? Here, ``weakness" either means a perturbation of small magnitude in the material
coefficients or a perturbation of finite magnitude, but small lateral extent. The latter
are particularly interesting for optical applications, since defects are usually
created by inserting materials with differing dielectric constant $\eps$ into the
photonic crystal structure.

Weak localization results are quite different from results for sufficiently
strong defects (like for example \cite{Hempel,FK97,FigKlein,KuchWav1,KuchWav2,MiaoMa}) and there are surprisingly few of them
in the literature. The first rigorous results on weak gap localization for periodic Schr\"odinger
operators were given by Parzygnat et al.~in \cite{Johnson}. Brown et al.~showed weak gap localization 
in \cite{BHPW15} for a periodic Helmholtz-type
operator corresponding to TM-mode polarization. 
We also refer to the paper of Parzygnat et al.~\cite{Johnson} for a thorough discussion of the literature 
on strong and weak localization for Schr\"odinger operators. In the slightly different setting of the coupling of two waveguides through narrow windows, weak localisation results for the Helmholtz equation were obtained in \cite{MP05,PTT10}.

For TE-polarized waves in periodic media, the problem is very challenging and we present, for the first time, conditions 
ensuring weak gap localization. The method we present here relies on \cite{BHPW15}, but the proofs are considerably more 
difficult due to the different structure of the operator. The chief difficulty is the fact that here the perturbation is
of the same order as the principal part of the differential operator. Moreover, we will be 
working with operators \eqref{eq:magfield} with non-smooth coefficients $\eps(\bx)$, 
requiring a sophisticated functional analytic setting.

%
\section{Problem statement}
%

We consider the propagation of electromagnetic waves in a non-magnetic, inhomogeneous medium
described by a varying dielectric function $\eps(\bX)$ with $\bX = (x, y, z)$. Assuming that
the magnetic field $\mathbf{H}$ has the form $\mathbf{H}= H(x, y) \hat{\bz}$, where $\hat{\bz}$ denotes the unit vector in the $z$-direction, we look for time-harmonic solutions to Maxwell's equations.  This leads to the equation
\beq \label{eq:magfield}
- \nabla \cdot \frac{1}{\eps(\bx)} \nabla H = \lambda H
\eeq
for the $\bz$-component $H$ of the magnetic field. Note that in the context of polarized waves, we
assume that all fields and constitutive functions depend only on $\bx = (x,y)$.
The periodic background medium is characterised by $\eps_0(\bx)$, where for simplicity we
assume that the unit square $[0, 1]^2$ is a cell of periodicity.

\subsection{Line defects}
 Let $\hat{\bx}=(1,0)$ and $\hat \by=(0,1)$. We now introduce a line defect, which we assume  to be aligned in the $\hat{\bx}$-direction and to preserve the periodicity in this direction. In addition, the defect is assumed to be localized in the $\hat \by$-direction.
The new system is therefore described by a  dielectric function $\eps_1(\bx)$, periodic in $\hat \bx$-direction, i.e.
\beq
\eps_1(\bx + m \hat{\mathbf{x}}) = \eps_1(\bx) \quad (m\in \Z)
\eeq
and there exists some $R>0$ such that
$\eps_1(\cdot,y)$ may differ from $\eps_0(\cdot,y)$, if $|y|< R$ and equals $\eps_0(\cdot,y)$ if $|y| > R$.

Since the system is still periodic in the $\hat{\bx}$-direction, we can apply Bloch's theorem \cite{OdehKeller,KuchmentBook}
to reduce our problem to a problem on the strip $\Omega:=(0,1)\times\R$. Thus, the generalized eigenfunctions of \eqref{eq:magfield}
have the form $e^{i k_x x} \psi^{(k_x)}(\bx)$, where $k_x\in[-\pi,\pi]$, $\psi^{(k_x)}$ is periodic in the $\hat\bx$-direction and satisfies
\beq\label{eq2}
- (\nabla + i k_x \hat\bx) \cdot\left[ \eps_1(\bx)^{-1} (\nabla + i k_x \hat\bx) \psi^{(k_x)}\right] = \lambda \psi^{(k_x)}.
\eeq
Equivalently, we may look for functions $u^{(k_x)}$ satisfying $k_x$-quasiperiodic boundary conditions
\beq\label{qp}
u^{(k_x)}(\bx + m \hat\bx) = e^{i k_x m} u^{(k_x)}(\bx)
\eeq
and solving the equation 
\beq\label{eq3}
- \nabla \cdot \left[\eps_1(\bx)^{-1} \nabla u^{(k_x)}\right] = \lambda u^{(k_x)}.
\eeq
For our purposes, it is slightly more convenient to use \eqref{eq3}, since unlike in \eqref{eq2}, the differential operator
is not changed. Note that the boundary condition \eqref{qp} now depends on $k_x$.

Suppose now $(\Lambda_0, \Lambda_1)$ is a band gap of the unperturbed system \eqref{eq:magfield} with $\eps=\eps_0$.
We will give conditions which ensure that localized modes appear in the interval below $\Lambda_1$
under arbitrarily weak perturbation. The unperturbed system is periodic with respect to two directions,
and the application of Bloch's theorem leads to the usual Bloch functions $\psi_s(\bx, k_x,k_y)$ and corresponding band
functions $\lambda_s(k_x, k_y)$ with $s\in\N$, $\bx\in[0,1]^2$ and $(k_x,k_y)\in[-\pi,\pi]^2$, see, e.g.~\cite{BHPW15} for more details. Let $M\in\N$ be such that  $\Lambda_1$ is the minimum
of the $M$-th band function and let $\bk^0=(k_x^0, k_y^0)$ be a value of the quasi momentum at which
$\Lambda_1$ is attained, i.e.
\beq\label{M}
\lambda_M(\bk^0) = \Lambda_1.
\eeq
 {We note that the minimum is attained; for more details see \cite[Proposition 3.2]{BHPW15}.}
For simplicity, we assume that $\lambda_M(k_x^0, k_y) \neq \Lambda_1$ for all
$k_y$ different from $k_y^0$.  We intend to deal with the more general case in forthcoming work.
 We note that due to analyticity of the function $k_y\mapsto\lambda_M(k_x^0, k_y)$ in a complex neighbourhood of the interval $[-\pi,\pi]$ (see e.g. \cite[Theorem VII.3.9]{Kato}),  we have 
 \beq\label{condBand}
\lambda_M(k_x^0, k_y) \leq \Lambda_1 + \alpha|k_y-k_y^0|^2
\eeq
close to $k_y^0$, for some $\alpha > 0$. (This also holds if $k_y^0=\pm\pi$, due to the periodic boundary behaviour of $\lambda_M(k_x^0,\cdot)$). 

One of the main features of this paper is that we do not require the functions $\eps_i$ to be continuous. The smoothness we require of the $\eps_i$ is merely that  $\eps_i\in L^\infty$. This is motivated by physical applications, where, to produce the typical band-gap spectrum,  $\eps_0$ is usually piecewise constant. See, for instance, \cite{CD99,FK96,Fil03}.
Moreover, we make the following assumptions on the perturbation: 
\begin{enumerate}
\item[(i)]
  $\eps_i\geq c_0>0$ for some constant $c_0$ and $i=0,1$.
\item[(ii)]The perturbation is nonnegative, i.e.~
\beq\label{positivity}
\eps_1(\bx) - \eps_0(\bx) \geq 0.
\eeq
\item[(iii)] There exists a ball $D$ such that $\eps_1-\eps_0>0$ on $D$. 
\end{enumerate}

We are now in a position to state our main result.

\begin{theorem}\label{theorem}  {In addition to (i), (ii) and (iii),} assume that
\beq\label{cond}
 \norm{\frac{\eps_1}{\eps_0}}_\infty \left\|\frac{1}{\eps_1}-\frac{1}{\eps_0}\right\|_\infty < \frac{\Lambda_1-\Lambda_0}{(\Lambda_0+1)}.
\eeq
Then weak localization takes place, i.e. the problem 
\beq\label{eq4}
- \nabla \cdot \eps_1(\bx)^{-1} \nabla u^{(k_x^0)} = \lambda u^{(k_x^0)},\quad \bx\in \Omega=(0,1)\times\R
\eeq
 has a nontrivial  $k_x^0$-quasiperiodic solution
$u^{(k_x^0)}\in L^2(\Omega)$ for some $\Lambda_0<\lambda<\Lambda_1$.
\end{theorem}

\begin{remark}
\begin{enumerate}
\item In fact, we have a slightly weaker condition for localization. \eqref{cond} can be replaced by
\beq\label{cond1}
 \left\|\frac{1}{\eps_1}-\frac{1}{\eps_0}\right\|_\infty < \frac{\Lambda_1-\Lambda_0}{\norm{\gp}_{H^{-1}\to H^1}(\Lambda_0+1)},
\eeq
where $\gp$ is the Green's operator introduced by \eqref{Gone}.
From the proof of Lemma \ref{lemma1}, we have $\norm{\gp}_{H^{-1}\to H^1}\leq\norm{\frac{\eps_1}{\eps_0}}_\infty$, thus \eqref{cond1} follows
from \eqref{cond}.
\item Condition \eqref{cond} is satisfied for sufficiently weak perturbations, so arbitrarily weak perturbations induce spectrum into the gap.
\item 
Note that $u^{(k_x^0)}\in L^2(\Omega)$ precisely expresses the type of localization we expect
in the context of line defects, i.e. the eigensolutions $u^{(k_x^0)}$ decay
in the direction perpendicular to the line defect, whereas they are
$k_x^0$-quasi periodic in the $\hat{\bx}$-direction. This is different from the
localization by point defects: these induce defect eigenfunctions that are
square integrable over the whole space.
\end{enumerate}
\end{remark}

%
\section{The periodic Green's function}
%

In this section, we recall the mathematical formalism needed and introduce the operators to be studied first in the $L^2$-setting.
Later on, we shall introduce realizations of the same operators in negative Sobolev spaces, required to apply the perturbation theory to 
nonsmooth coefficients.

The unperturbed operator $L_0$  is defined  in a standard way using the representation theorem (see \cite{Kato})
from the sesquilinear form
\bea
\int_\Omega \frac{1}{\eps_0(\bx)} \nabla u\overline{\nabla v} d\bx 
\eea
where $u, v$ are $H^1$-functions on $\Omega$ satisfying $k_x$-quasiperiodic boundary conditions in $\hat \bx$-direction.
The pertubed operator $L_1$ is defined in a similar way, replacing $\eps_0$ by $\eps_1$.

As our technique is based on exploiting the Bloch representation of the Green's functions (or, equivalently, the resolvent operators), combined with
a variational approach, we first review the definition and properties of the Green's function.

Central to our analysis is the Green's function $G_0(\bx, \bx')$ (see e.g. \cite{Eco})
satisfying
\beq\label{Green}
(L_0 + 1)~~G_0(\bx, \bx') = \delta(\bx - \bx').
\eeq
We note that the Green's function in \eqref{Green} is then also subject to  $k_x^0$-quasiperiodic
boundary conditions:
\beq
G_0(\bx + m \hat \bx, \bx') = e^{i k_x^0 m} G_0(\bx, \bx')
\eeq
for all integers $m$. It is very useful to have a representation of $G_0(\bx, \bx')$ in terms of
eigenfunctions, i.e. Bloch waves. We shall now derive such a representation.

The set of Bloch waves $\psi_s(\bx, \bk)$ is known to form a complete system in the space of
square-integrable functions defined on the whole space. Likewise, the Bloch functions
$\psi_s(\bx, k_x^0, k_y)$ with the $x$-component of the quasimomentum fixed, form a complete system
in the space of square-integrable functions on the strip $\Omega$. This means that any
such function $f$ can be expanded in terms of Bloch waves:
\beq\label{repf}
f(\bx) = \frac{1}{\sqrt{2\pi}} \sum_s \int_{-\pi}^{\pi} \la U f(\cdot,k_y), \psi_s(\cdot, k_x^0,k_y) \ra \psi_s(\bx,k_x^0, k_y) ~~d k_y
\eeq
where  $U$ denotes the Floquet transform in the $\hat\by$-direction and the series converges in the $L^2$-sense. Here, $\la U f(\cdot,k_y), \psi_s(\cdot,k_x^0, k_y) \ra$ is the $L^2$-inner product over the unit square $[0,1]^2$.
Note that in \eqref{repf}, we only integrate over the $k_y$-component of the quasi-momentum.
To simplify notation, we will also write $\lambda_s(k_y):=\lambda_s(k_x^0,k_y)$ and $\psi_s(\bx, k_y)=\psi_s(\bx, k_x^0, k_y)$ in the following.

As in \cite[Chapter 1]{Eco}, \eqref{repf} immediately implies the following representation:
\ben
G_0(\bx, \bx') =
\frac{1}{2\pi} \sum_{s} \int_{-\pi}^{\pi} \frac{\psi_s(\bx', k_y) \psi_s(\bx, k_y)}{\lambda_s(k_y)+1}~~d k_y. \label{Blochrep}
\een
Formula \eqref{Blochrep} is extremely powerful. As we shall show, it allows us to analyze rigorously
the interaction of the defect with the Bloch waves of the unperturbed system.
It is convenient to write $G_0 := (L_0+1)^{-1}$, i.e.
$$
(G_0 f)(\bx) = \int_\Omega G_0(\bx, \bx') f(\bx')~d \bx'.
$$
 {Then $G_0$ is a symmetric and positive operator in $L^2(\Omega)$.} 
We also need to introduce the analogous Green's operator
for $L_1$ (subject
to $k_x^0$-quasiperiodic boundary conditions in the $\hat \bx$-direction).
Since the spectrum of the differential operator $L_1$ is contained in
the positive half-axis, $G_1:=(L_1+1)^{-1}$ exists  { and is a symmetric positive operator in $L^2(\Omega)$.}
Note that the coefficients of $L_1$ describe the perturbed
system and have no periodicity in the $\hat \by$-direction. As a consequence, $G_1$ cannot
be expressed in terms of Bloch waves, as in \eqref{Blochrep}.

We will show that the essential spectra of $L_0$ and $L_1$ coincide, so any new spectrum introduced in the gap can only consist of eigenvalues.
The key idea of our approach is then to transform the eigenvalue problem \eqref{eq4}
into an eigenvalue problem for the Green's operator $G_1$.
In fact, suppose that $u=u^{(k_x^0)}\neq 0$ solves
\eqref{eq4} together with the boundary condition
\eqref{qp}, i.e. $$(L_1 - \lambda) u = 0.$$ It is easy to see that this is equivalent
to
\beq\label{eq5}
G_1 u = \mu u,
\eeq
where
\beq\label{eqmu}
\mu = (\lambda+1)^{-1}.
\eeq
Thus the eigenvalue problem consisting of \eqref{eq4} and \eqref{qp} can be transformed into the eigenvalue
problem \eqref{eq5}, with $\lambda$ and $\mu$ related by \eqref{eqmu}.
For $\lambda\in(\Lambda_0, \Lambda_1)$ we have
$\mu\in((\Lambda_1+1)^{-1}, (\Lambda_0+1)^{-1})$.
We may apply the same reasoning to $L_0$ yielding a similar relation between the spectra of $L_0$ and $G_0$. 

%
\section{The operator theoretic formulation}
%

In this paper we shall only assume $\eps_0,\eps_1\in L^{\infty}$ with a positive lower bound. In order to deal with the lack of smoothness in the coefficients, we will work in negative Sobolev spaces. In particular, rather than study the operators $G_0$ and $G_1$ directly, we will consider their realisations in the space $H^{-1}_{qp}(\Omega)$, introduced below, denoted by $\go$ and $\gp$, respectively. 

 Recall that we
work with quasimomentum $k_x^0$ fixed. To construct the $H^{-1}$-realisations, we introduce the space of quasi-periodic 
$H^1$-functions on $\Omega$
$$
H^1_{qp}(\Omega):=
\{ u \in H^1_\text{loc}(\R^2) :u\vert_{\Omega}\in H^1(\Omega) \hbox{ and } u(\bx+(m,0))=e^{i k_x^0 m} u(\bx), m\in \Z, \bx\in \R^2 \}. 
$$

For $u,v\in H^1_{qp}(\Omega)$ consider the sesquilinear form 
$$B_0 [u,v]=\int_\Omega \left(\frac{1}{\eps_0(\bx)}\nabla u \overline{\nabla v} + u\overline{v}\right)~d\bx.$$
As $\eps_0$ is bounded and bounded away from zero, we can
introduce a new scalar product on 
$H^1_{qp}(\Omega)$ given by $$ \la u,v\ra
_{H^1_{qp}(\Omega)}:=B_0[u,v]$$ which is equivalent to the usual
scalar product in  $H^1(\Omega)$. When there is no danger of confusion, we denote the associated norm
$\norm{\cdot}_{H^1}$.

\begin{defn}
Let $H^{-1}_{qp}(\Omega)$ denote the dual space of $H^1_{qp}(\Omega)$. Let $\phi:H^1_{qp}(\Omega)\to H^{-1}_{qp}(\Omega)$ be defined by
\begin{equation}
\llangle\phi[u],\varphi\rrangle =B_0[u,\varphi]  \mbox{  \; for\; all\; } u,\varphi\in H^1_{qp}(\Omega)
\label{eq:star}
\end{equation}
where the $\llangle\cdot,\cdot \rrangle $-notation indicates the dual pairing, i.e. $\llangle w, \varphi\rrangle$ is the action of the linear functional $w$
on the function $\overline \varphi$. We shall also use
$
w[\varphi]
$
to denote the dual pairing.
\end{defn}
$\phi$ is an isometric isomorphism, and hence  the scalar product on $H^{-1}_{qp}(\Omega)$ given by $$\llangle u,v\rrangle _{H^{-1}_{qp}(\Omega)}:=\llangle \phi^{-1}u,\phi^{-1}v\rrangle _{H^1_{qp}(\Omega)}$$ induces a norm which coincides  with the usual operator sup-norm on $H^{-1}_{qp}(\Omega)$.
 
 {We next introduce the realisations of $L_0$ and $G_0$ in $H^{-1}_{qp}(\Omega)$.}

\begin{prop}\label{prop2.1} We define an operator $\lo:D(\lo)\to H^{-1}_{qp}(\Omega)$ by
$D(\lo):=H^1_{qp}(\Omega)\subset H^{-1}_{qp}(\Omega)$ and
$$\lo u:=\phi u - u \label{eq:2.1}.$$
Then $\lo$ and $\go:=(\lo+1)^{-1}$ are  self-adjoint.
\end{prop}

\begin{proof}
For $u,v\in H^1_{qp}(\Omega)$,
\bea 
   \llangle (\lo+1) u,v\rrangle _{H^{-1}} &={}& \llangle \phi^{-1}(\lo+1) u,\phi^{-1}v\rrangle _{H^1} \
   =\ \llangle u,\phi^{-1}v\rrangle _{H^1}\ =\ \overline{ \llangle \phi^{-1}v,u\rrangle _{H^1}}\\
	&=&\overline{ B_0[\phi^{-1}v,u]} 
   \ =\ \overline{ \llangle v,u\rrangle }\ =\ \overline{\llangle v,u\rrangle _{L^2}}\ =\ \llangle u,v\rrangle _{L^2};
\eea
the last line follows by (\ref{eq:star}). Thus  $\lo+1$ is symmetric.

Since $\phi$ is bijective it follows   that $\lo+1$ is bijective, thus
$(\lo+1)^{-1}:H^{-1}_{qp}(\Omega)\to H^{-1}_{qp}(\Omega)$ is defined  on the whole space  and is also symmetric. Therefore, $\go=(\lo+1)^{-1}$ is self-adjoint. Hence $\lo+1$, and so $\lo$ itself  is self-adjoint.
\end{proof}

\begin{remark}\label{remark:2}
\begin{enumerate}
	\item The map $\phi$ corresponds to the operator $\lo+1$ and  $\phi^{-1}:H^{-1}_{qp}(\Omega)\to H^1_{qp}(\Omega)$ acts in the same way as $\go :H^{-1}_{qp}(\Omega)\to H^{-1}_{qp}(\Omega)$.  
	\item We remind the reader of the standard embedding of $L^2(\Omega)$ in $H^{-1}_{qp}(\Omega)$:
       a function $f\in L^2(\Omega)$ acts on $v \in H^1_{qp}(\Omega)$ via  
$f[v] = \llangle f,v\rrangle _{L^2}$.
  \item  From the definitions of $\phi$ and $\lo $ follows the useful identity
        $$ \llangle u, v\rrangle _{H^{-1}} = \llangle \go u, \go v\rrangle _{H^1} = \llangle u, \go v\rrangle _{L^2}\quad \hbox{for} \quad u\in L^2(\Omega), v\in H^{-1}_{qp}(\Omega).$$
	\item We note that just as in \cite[Section 5]{BHPW11}, the $L^2$- and $H^{-1}$-spectra coincide:
\bea
\sigma(L_0) = \sigma(\lo). 
\eea
\end{enumerate}
\end{remark}

 {Let $\mu\in ((\Lambda_1+1)^{-1}, (\Lambda_0+1)^{-1})$. Then by the previous remark, $1/\mu\in \rho(\lo+1)$, so
$(I-\mu (\lo +1))^{-1}=(I-\mu \go^{-1})^{-1}$ is well defined and maps $H^{-1}_{qp}(\Omega)$ bijectively onto $H^1_{qp}(\Omega)$.
The operator
$(I-\mu \go^{-1})^{-1}$ is the solution operator to the problem 
$$ \llangle u,\varphi\rrangle_{L^2} -\mu \int_{\Omega} \left (\tfrac{1}{\eps_0} \nabla u \overline{\nabla \varphi} + u \overline\varphi\right)d\bx = f[\varphi], \quad \hbox{ for all } \varphi \in H^1_{qp}(\Omega)
$$
for a given $f\in H^{-1}_{qp}(\Omega)$.}

We now introduce the solution operator for the perturbed problem. 
Let $\gp$ be the operator defined on $H^{-1}_{qp}(\Omega)$ such that for given $f\in H^{-1}_{qp}(\Omega)$ the function $u=\gp f$ is the unique solution in $ H^1_{qp}(\Omega)$ to
\beq
\label{Gone} B_1[u,\varphi]:=\int_\Omega \left [  \dfrac{1}{\eps_1} \nabla u   \overline{\nabla \varphi} + u \overline\varphi \right ] dx =f[\varphi]\quad\hbox{ for all } \varphi \in H^1_{qp}(\Omega).
\eeq

We see that
 $\gp$ is well-defined, since it can be constructed via a form in the same way as $\go$, noting that the norms in the $H^{-1}_{qp}$-spaces constructed from both sesquilinear forms $B_0$ and $B_1$ are equivalent.
	Note that $\go\vert_{L^2}=G_0$ and $\gp\vert_{L^2}=G_1$, which are both  symmetric operators in $L^2$. Moreover,  again, as in \cite[Section 5]{BHPW11}, the $L^2$- and $H^{-1}$-spectra coincide: $\sigma(G_1) = \sigma(\gp)$. We also denote $\lp=\gp^{-1}-1$.

We conclude the section with the proof of some more simple properties of $\gp$ and $\go$ which will be useful later.  {Recall that, by assumption, $\eps_1\geq\eps_0$.}

\begin{lemma} \label{lemma1}
$\gp:H^{-1}_{qp}(\Omega)\to H^1_{qp}(\Omega)$ is bounded with  {$\norm{\gp}_{H^{-1}\to H^1}\leq\norm{\frac{\eps_1}{\eps_0}}_\infty$.}
\end{lemma}

\begin{proof}
Let $f \in H^{-1}_{qp}(\Omega)$ and $u=\gp f$. Choose $\varphi=u$ in \eqref{Gone}. It then follows that
 {
\bea 
   \norm{u}_{H^1}^2 & = & \int \left ( \dfrac{1}{\eps_0}  |\nabla u |^2 + |u|^2\right ) dx \\
  & \leq  & \norm{\frac{\eps_1}{\eps_0}}_\infty\int   \left (\dfrac{1}{\eps_1}  |\nabla  u|^2 + |u|^2\right   )  dx =\norm{\frac{\eps_1}{\eps_0}}_\infty f[u] \leq \norm{\frac{\eps_1}{\eps_0}}_\infty\norm{f}_{H^{-1} }\norm{u}_{H^1},
\eea
where  $\norm{\frac{\eps_1}{\eps_0}}_\infty$ is bounded by assumption.}
\end{proof}

\begin{lemma} \label{lem3}
For $w \in H^{-1}_{qp}(\Omega), \;\;\; w[\mathfrak{G}_i w]\geq 0$ for $i=0,1$.
\end{lemma}
\begin{proof}
Choose a sequence $(w_n)\in (L^2(\Omega))^\N$ such that $w_n \to w$ in $H^{-1}_{qp}(\Omega)$.
By continuity  of $\mathfrak{G}_i :H^{-1}_{qp}(\Omega)\to H^1_{qp}(\Omega)$, we have $\mathfrak{G}_i w_n\to \mathfrak{G}_i  w$ in  $H^1_{qp}(\Omega)$, so $w_n[\mathfrak{G}_i  w_n]\to w[\mathfrak{G}_i w]$.
Furthermore,  $$w_n[\mathfrak{G}_i  w_n]=\int w_n \overline{G_i  w_n} \geq 0,$$ since $G_i \geq 0$ as operators in $L^2.$
\end{proof}

%
\section{Birman-Schwinger-type reformulation}
%

An essential feature of our approach is to first perform a Birman-Schwinger-type reformulation of the problem. 
In this way, we bring the unperturbed Green's operator into play. We will show below (see Lemma \ref{lemmaes}) that $\gp-\go$ is 
compact as an operator in $H^{-1}_{qp}(\Omega)$. Hence the spectra of $\go$ and $\gp$ can only differ by eigenvalues.

The eigenvalue problem for our original operator, $(L_1-\lambda) u=0$ with $\lambda\in(\Lambda_0,\Lambda_1)$, is equivalent to
\beq
(\gp -\mu)u=0,\quad u\in   H^{-1}_{qp}(\Omega)  \label{star}
\enq
for $\mu=(\lambda+1)^{-1}\in ((\Lambda_1+1)^{-1}, (\Lambda_0+1)^{-1})$. Then \eqref{star} implies that each eigenfunction $u$ lies in $H^1_{qp}(\Omega)$ and so
\ben  \nonumber
 (\gp -\mu)u=0 &\Leftrightarrow & (\go -\mu)u + (\gp -\go )u = 0 \\ \nonumber
 &\Leftrightarrow & (I-\mu \go ^{-1})u + ( \go ^{-1}\gp -I)u = 0 \\
 &\Leftrightarrow &  u + (I-\mu \go ^{-1})^{-1} ( \go ^{-1}\gp -I)u = 0, \label{reform1}
\een
where the last equivalence follows, as  {$(I-\mu \go ^{-1})^{-1}: H^{-1}_{qp}(\Omega)\to H^1_{qp}(\Omega)$ 
} is bijective. We are therefore interested in the operator $$(I-\mu \go ^{-1})^{-1} ( \go ^{-1}\gp -I).$$ We first study 
$\go ^{-1}\gp -I$.

\begin{lemma} \label{lemma2}
Let $K=(\go ^{-1}\gp -I):H^{-1}_{qp}(\Omega) \to H^{-1}_{qp}(\Omega)$. Then $K$ is symmetric.
\end{lemma}

\begin{proof} Let $u,v\in H^{-1}_{qp}(\Omega)$ and choose
 sequences $(v_n)$ and $(u_m)$ in $L^2(\Omega)$ such that
$v_n\to v$   and $u_m\to u$ in $H^{-1}_{qp}(\Omega)$.
We first note that
\ben \label{symmK} 
   \llangle K u,v\rrangle_{H^{-1}}& = &  \llangle \phi^{-1} K  u, \phi^{-1} v \rrangle_{H^1}=\llangle \go K  u, \phi^{-1} v \rrangle _{H^1} \\
  &  = &\llangle (\gp -\go )u,\phi^{-1}v\rrangle_{H^1}\ =\ \overline{v [ (\gp -\go )u]}. \nonumber
\een	
Now, using the convergence in $H^{-1}_{qp}(\Omega)$ and the symmetry of the $G_i$, $i=0,1$, in $L^2(\Omega)$  we get
\bea	v [ (\gp -\go )u] &=& \lim_{n\to \infty} v_n[(\gp -\go )u]
    \  = \ \lim_{n\to\infty} \llangle v_n, (\gp -\go )u\rrangle  _{L^2} \\
     &  = &  \lim_{n\to\infty}\lim_{m\to \infty} \llangle v_n, (G_1-G_0)u_m\rrangle_{L^2}\\
      &  =&  \lim_{n\to\infty}\lim_{m\to \infty} \llangle  (G_1-G_0)v_n, u_m \rrangle_{L^2}\\
       &  = &  \lim_{n\to\infty}\lim_{m\to \infty}   \overline{u_m[(G_1-G_0)v_n]} \\
        &   =&  \lim_{n\to\infty} \overline{u[(\gp -\go)v_n]}=\overline{u[(\gp -\go)v]}. 
\eea
By a similar calculation to \eqref{symmK}, this equals $ \llangle  K v,u\rrangle_{H^{-1}}=\overline{ \llangle u,Kv\rrangle}_{H^{-1}}$, proving the result.
\end{proof}

\begin{lemma}
For $u \in H^{-1}_{qp}(\Omega),\;\; \llangle Ku,u\rrangle_{H^{-1}}\geq 0$.
\end{lemma}

\begin{proof}
From the proof of Lemma \ref{lemma2} we have
$$\llangle Ku,u\rrangle_{H^{-1}}= u[(\gp -\go )u].$$

Moreover, $(\lo+1)\go u=u \in H^{-1}_{qp}(\Omega)$ and 
\bea 
 \left (  (\lp+1)\go u\right ) [\gp u]&=&  \int_{\Omega} \left(\frac{1}{\eps_1} \nabla \go  u\overline{\nabla \gp u}+\go u\overline{\gp u}\right)d\bx \ =\ 
\overline{\left ( (\lp+1)\gp u\right ) [\go u]}\\
&=& \overline{u[\go u]}\ =\ \overline{\llangle \go  u,\go u\rrangle}_{H^1}\ = \ \llangle \go  u,\go u\rrangle_{H^1}\ = \ u[\go u].
\eea
Combining these three equalities, we get
\bea 
   \llangle Ku,u \rrangle_{H^{-1}}&=& u[(\gp -\go )u]\ =\ u[\gp u] -u[\go u]\\
	&=& ((\lo+1)\go u)[\gp u] -  ((\lp+1)\go u) [\gp u]\\
	&= & -\left( (\lp-\lo)\go  u \right ) [ \gp u] \\	
   &=&  ( (\lp-\lo) \go  u )   [ (\gp  (\lp-\lo)\go -\go )u ] 
   \\
   &=& \left ( (\lp-\lo) \go  u \right )[\gp  (\lp-\lo)\go  u] -\left ( (\lp-\lo)\go u \right ) [\go u]. \\
\eea
The first term is non-negative by Lemma \ref{lem3}.
Also,
$$
-\left ( (\lp-\lo)\go  u\right ) [\go u]=\int_\Omega \left  (   \frac{1}{\eps_0}-\frac{1}{\eps_1} \right ) | \nabla \go  u|^2 \geq 0,$$
implying $\llangle Ku,u\rrangle_{H^{-1}}\geq 0$.
\end{proof}

\begin{lemma}\label{lem6}
$(I-\mu \go ^{-1})^{-1}$ is symmetric in $H^{-1}_{qp}(\Omega)$ for $\mu\in((\Lambda_1+1)^{-1}, (\Lambda_0+1)^{-1})$.
\end{lemma}

\begin{proof}
%
This is obvious, since $\lo $ is a self-adjoint operator in $H^{-1}_{qp}(\Omega)$.
    \end{proof}

 In order to proceed, it is most convenient to modify the equation
\eqref{reform1} by suitably projecting out the null space of $K$.
 Set $\ck = \overline{\Ran K}$  and let  $P:H^{-1}_{qp}(\Omega)\to \ck$ be the orthogonal projection.
On $\ck$, we introduce a new inner product given by
\beq\label{defInnerprod}
\la f, g \ra_\ck  := \la K f, g \ra_{H^{-1}}.
\eeq 
We first show  the definiteness of this inner product.

\begin{lemma} \label{lem5}
$\la\cdot,\cdot\ra_\ck $ is positive definite on  $\ck$.
\end{lemma}
\begin{proof} 
Suppose  $\la u,u\ra_\ck=\llangle Ku,u\rrangle_{H^{-1}}=0$ for some $u\in\ck$.  As $K\geq 0$, we can define $K^{1/2}$ as a selfadjoint operator in $H^{-1}_{qp}(\Omega)$ and get $K^{1/2}u=0$,
implying that  $Ku=0$.
Thus $u\in\ck \cap \ker K$ giving $u=0$.
%
\end{proof}

Furthermore, we have the following bounds.

\begin{lemma} \label{lem12} We have the estimates 
$$ (i) \quad \norm{ Ku}_{H^{-1}}\leq \norm{\gp  }_{H^{-1}\to H^1} \norm{\frac{1}{\eps_0}-\frac{1}{\eps_1}}_\infty \norm{u}_{H^{-1}},$$
 $$ (ii)\quad \norm{K}\leq  \norm{\gp  }_{H^{-1}\to H^1} \norm{\frac{1}{\eps_0}-\frac{1}{\eps_1}}_\infty,$$     
$$ (iii) \quad \norm{ Ku}_{H^{-1}}^2\leq \norm{K}\norm{u}^2_\ck.$$
Moreover, if $\delta:=\norm{\frac{1}{\eps_0}-\frac{1}{\eps_1}}_\infty<1/\norm{\go  }_{H^{-1}\to H^1}$, then
$$ (iv) \quad \norm{\gp  }_{H^{-1}\to H^1} \leq \frac{\norm{\go  }_{H^{-1}\to H^1}}{1-\delta \norm{\go }_{H^{-1}\to H^1}}.$$
\end{lemma}

\begin{proof}
 The identity $K=(\lo +1) \gp  -I = (\lo -\lp) \gp  $
   implies for $u\in H^{-1}_{qp}(\Omega)$ and $\varphi\in H^1_{qp}(\Omega)$,
    \beq 
      Ku[\varphi]= \int _\Omega  \left( \frac{1}{\eps_0}-\frac{1}{\eps_1} \right) \nabla \gp u \nabla \overline\varphi.
		\eeq	
			Therefore,
   \beq  | Ku[\varphi]|\leq \norm{\frac{1}{\eps_0}-\frac{1}{\eps_1}}_\infty \norm{\gp  u}_{H^1} \norm{\varphi}_{H^1},\eeq 
		proving (i) and (ii).
		
      Since $K\geq 0$,  
      $$\norm{K}=\sup \dfrac{ \llangle Ku,Ku\rrangle}{\llangle u, Ku\rrangle}=\sup\frac{\norm{Ku}_{H^{-1}}^2}{\norm{u}_\ck^2}.$$
        Thus we have (iii).
		
				Finally, as $\gp-\go=\go K$, using (ii) we have 
				$$\norm{\gp  }_{H^{-1}\to H^1} \leq {\norm{\go  }_{H^{-1}\to H^1}}(1+\norm{K})\leq  \norm{\go  }_{H^{-1}\to H^1}(1+\delta \norm{\gp  }_{H^{-1}\to H^1} )$$
				and rearranging gives the desired inequality.
   \end{proof}

Note in particular that this means that for small perturbations, the only dependence of the bound for $\norm{Ku}_{H^{-1}}$ on the perturbation $\eps_1$ is through the term $\norm{\frac{1}{\eps_0}-\frac{1}{\eps_1}}_\infty $.

We now introduce the operator we wish to study. Let
\beq A_\mu:= P(I-\mu \go^{-1})^{-1} K :\ck\to\ck. \eeq

\begin{lemma}\label{lemmavar}
Equation \eqref{star} has a non-trivial solution $u$ iff $-1$ is an eigenvalue of $A_\mu$.
\end{lemma}

\begin{proof}
Let $v = P u$,  {where $u$ is a solution of \eqref{star}}.    Applying $P$ to \eqref{reform1} shows that $v$ solves
\beq\label{reform3}
v + A_\mu v = 0.
\eeq
Conversely, one easily checks that a solution $v\neq 0$ of \eqref{reform3}
gives a solution $u\neq 0$ of the original problem \eqref{star}: we just have
to set $u =  - (I-\mu\go^{-1})^{-1}Kv$ and use \eqref{reform1}, noting that $Pu=v$,  {so $Ku=Kv$}.
%
\end{proof}

\begin{lemma}\label{lem7}
$A_\mu$ is symmetric in $\ck$.
\end{lemma}
\begin{proof}
Let $u,v\in \ck$. Then
\bea 
   \llangle A_\mu u,v\rrangle _\ck  & = & \llangle KA_\mu u,v\rrangle_{H^{-1}}=\llangle A_\mu u, Kv\rrangle_{H^{-1}} 
   \ =\  \llangle P(I-\mu \go^{-1})^{-1}Ku,Kv\rrangle_{H^{-1}}\\
   &=& \llangle (I-\mu \go^{-1})^{-1}Ku,Kv\rrangle_{H^{-1}} \
=\ \llangle Ku,(I-\mu \go^{-1})^{-1}Kv\rrangle_{H^{-1}}\
=\ \llangle u, A_\mu v\rrangle _\ck,
\eea
where we have used Lemma \ref{lemma2} and Lemma \ref{lem6}.
\end{proof}

In the following recall that $\eps_1 = \eps_0$ if $|y| > R$.
\begin{lemma}\label{cs}
Let $H^{-1}_{cs}$ denote the space of distributions in $H^{-1}_{qp}(\Omega)$ with compact support in the $\hat{\by}$-direction, the support being
contained in $[0,1]\times [-R, R]$. Then $\Ran K\subseteq H^{-1}_{cs}$.
\end{lemma}

\begin{proof}
Let $f\in \Ran K$, i.e.~there exists $g\in\ H^{-1}_{qp}(\Omega)$ such that $f=Kg=(\go^{-1}\gp -I)g$. Then for any $\varphi\in H^1_{qp}(\Omega)$ we have
\bea
f[ \varphi] &=& (\go^{-1}\gp  g)[\varphi] - g[\varphi] \ =\ \int_\Omega\left( \frac{1}{\eps_0} \nabla \gp g \overline{\nabla\varphi} + \gp g\overline{\varphi}\right)- g[\varphi]\\
&=& \int_\Omega\left( \frac{1}{\eps_0}-\frac{1}{\eps_1}\right) \nabla \gp g \overline{\nabla\varphi} + \int_\Omega \left( \frac{1}{\eps_1} \nabla \gp g \overline{\nabla\varphi}+\gp g\overline{\varphi}\right)- g[\varphi]\\
&=& \int_\Omega\left( \frac{1}{\eps_0}-\frac{1}{\eps_1}\right) \nabla \gp g \overline{\nabla\varphi}.
\eea
Observe that the
second integral
in the second line of the calculation is just $\left((\lp + 1)\gp g\right)[\varphi]$.
It is therefore clear that $f$ vanishes on all functions $\varphi$ supported outside the support of $\frac{1}{\eps_0}-\frac{1}{\eps_1}$.
\end{proof}
 We are now in a  position  to establish that the spectrum of $A_\mu$ consists only of eigenvalues.  We first show that $\gp-\go$ enjoys the same property.

 \begin{lemma} \label{lemmaes}
The essential spectra of  $\go$ and $\gp$ coincide.
\end{lemma}

\begin{proof}
We shall  show that $\gp-\go$  is compact as an operator in $H^{-1}_{qp}(\Omega)$. We have the following mappings:
$$ \gp-\go=\go { K} : {H^{-1}_{qp}(\Omega)}   \overset  { K} \to H^{-1}_{cs}   \overset { \go}   \to   H^1_{qp}  (\Omega, e^{\gamma | y|})  \overset{c}  \hookrightarrow H^{-1}_{qp}(\Omega), $$
where $\gamma>0$ and $H^1_{qp}  (\Omega, e^{\gamma | y|}) $ is the space of functions $u\in H^1_{qp}(\Omega)$ such that $e^{\gamma |y|} u(x,y)\in H^1_{qp}(\Omega)$ with norm $\norm{u}_{H^1_{qp}  (\Omega, e^{\gamma | y|})}:=\norm{e^{\gamma | y|}u}_{H^1}$. The last embedding is compact (see the Appendix). It remains to show that $\go:H^{-1}_{cs} \to H^1_{qp}(\Omega, e^{\gamma | y|}) $, and that it is bounded.
Let $f\in H^{-1}_{cs}$ and $u=\go f$,  i.e.~for all $\varphi \in H^1_{qp}(\Omega)$ we have
$$ f[\varphi] =\int_\Omega\frac1{\eps_0} \nabla u \overline{\nabla \varphi} + u \overline\varphi.$$
Set $\omega= e^{\gamma | y |} u$ and $\varphi= e^{\gamma | y|}\psi$ where $\psi$ is compactly supported.
Then 
\ben \nonumber 
f (e^{\gamma |\cdot|}\psi)&= & \int_\Omega\frac1{\eps_0} \left ( \nabla -\gamma \frac y {|y|}  {\bf{\hat y}}\right )\omega \left( \nabla +\gamma \frac y {|y|}  {\bf{\hat y}}\right ) \overline\psi +\omega \overline\psi 
\ =\ \int_\Omega  \left ( \frac1{\eps_0} \nabla \omega \overline{\nabla \psi} + \omega \overline\psi \right )   + \gamma S_\gamma \omega [\psi]  
\een
where we have set $S_\gamma: H^1_{qp}(\Omega) \to H^{-1}_{qp}(\Omega)$,
\ben\label{Sgamma} S_\gamma \omega [\psi] :=  \int_\Omega\frac{1}{\eps_0} 
 \left(-\frac{y}{|y|} \overline{\frac{\partial\psi}{\partial y}} \omega+\frac{y}{|y|} \frac{\partial \omega}{\partial y} \overline\psi -\gamma \omega\overline\psi\right).\een
Now, $f\circ e^{\gamma |\cdot |}=(\lo +1)\omega + \gamma S_\gamma \omega \in H^{-1}_{cs}$, as $f$ is. Hence, $\go(f \circ e^{\gamma |\cdot|}) =\omega + \go (\gamma S_\gamma \omega)=(I+\go\gamma S_\gamma)\omega$.
For small $|\gamma|$ this can be inverted by the Neumann series, so  $\omega
=(I+\go\gamma S_\gamma)^{-1}\go (f\circ e^{\gamma |\cdot|})$ and
$$\norm{\omega}_{H^1} \leq \norm{(I+\go\gamma S_\gamma)^{-1}}_{H^1\to H^1} \norm{\go(f\circ e^{\gamma |\cdot|})}_{H^1}
$$
Thus  $\norm{e^{\gamma |\cdot|} u}_{H^1} 
\leq C_\gamma \norm{f}_{H^{-1}}.$
\end{proof}


\begin{prop}  \label{AmuC}
$A_\mu:\ck\to\ck$ is compact.
\end{prop}

\begin{proof}
We rewrite the operator as  $A_\mu= P \go^{-1}(I-\mu \go^{-1})^{-1} \go {K}$. By Lemma \ref{lem12}, the operator $K:{\mathcal K} \to H^{-1}_{qp}(\Omega)$ is bounded {, and it maps into $H^{-1}_{cs}$ by Lemma \ref{cs}.} Since, again using Lemma \ref{lem12}, we have
$$ \norm{Pu}_\ck^2 = \llangle KPu,Pu\rrangle_{H^{-1}}=\llangle Ku,u\rrangle_{H^{-1}}\leq \norm{Ku}_{H^{-1}} \norm{u}_{H^{-1}}\leq C \norm{u}_{H^{-1}}^2,$$
 the operator $P:H^{-1}_{qp}(\Omega)\to {\mathcal K}$ is bounded. As $\mu^{-1}\in\rho(\lo +1)$, the operator $I-\mu(\lo+1):H^1_{qp}(\Omega)\to H^{-1}_{qp}(\Omega)$ is onto, and bounded, and hence also $(  I-\mu \go^{-1})^{-1}:H^{-1}_{qp}(\Omega)\to H^1_{qp}(\Omega)$ is continuous and,  {as in the proof of Lemma \ref{lemmaes},}
we have the following mapping properties
$$ {\mathcal K} { \overset K \to H^{-1}_{cs}  } \underset { compact}{  \overset {\go}   \to }  H^{-1}_{qp}(\Omega)  
 {  \overset {(  I-\mu \go^{-1})^{-1}}   \longrightarrow }   H^1_{qp}(\Omega)
 \overset { \go^{-1}} \to H^{-1}_{qp}(\Omega)  \overset { P} \to {\mathcal K}.
$$
Thus, $A_\mu:\ck\to\ck$ is compact.
 \end{proof}

%
\section{Existence of spectrum for weak perturbations}\label{mainProof}
%

We next estimate the eigenvalues of $A_\mu$ using variational methods.
%
Lemma \ref{lemmavar} enables us to study our spectral problem by  applying variational methods to the equation \eqref{reform3}.
As a mathematical subtlety, note that
$\ck$ is in general not complete with $\la\cdot, \cdot \ra_\ck $ as an inner product.
However, this does not affect our arguments, since the spectral theory of symmetric compact
operators is applicable on Pre-Hilbert spaces (see \cite{Heuser}).

It follows from our analysis below that (at least) for some $\mu$ in the spectral gap $\left((\Lambda_1+1)^{-1},(\Lambda_0+1)^{-1}\right)$, the operator
$A_\mu$ has a negative eigenvalue.  
Our strategy consists in following ${\kappa}(\mu)$, the most negative eigenvalue of the
operator $A_\mu$, as $\mu$ varies.
${\kappa}(\mu)$ can be characterized by
\beq\label{defkappa}
{\kappa}(\mu) = \min_{u\neq 0} \frac{\la u, A_\mu u\ra_\ck }{\la u, u \ra_\ck }.
\eeq
We prove below that ${\kappa}(\mu)$ is monotonically increasing in $\mu$
and continuous, in the range $(\Lambda_1+1)^{-1} < \mu < (\Lambda_0+1)^{-1}$, 
and that ${\kappa}(\mu)$ goes to
$-\infty$ as $\mu$ approaches $(\Lambda_1+1)^{-1}$ from the right. At the same time, we will
find a $\tilde \mu$ to the right of $(\Lambda_1+1)^{-1}$ for which $-1 < {\kappa}(\tilde \mu)$,
provided \eqref{cond} is satisfied. Hence ${\kappa}(\mu)=-1$ holds necessarily for some
$\mu$, i.e. $A_\mu$ has $-1$ as an eigenvalue and \eqref{reform3} has a non-trivial solution.

\begin{lemma}\label{lem:cont} For $\mu$ in the spectral gap $\left((\Lambda_1+1)^{-1},(\Lambda_0+1)^{-1}\right)$ we have that
$\mu\mapsto \kappa(\mu)$ is continuous and increasing.
\end{lemma}

\begin{proof}
As $\mu\mapsto A_\mu$ is norm-continuous, we have that  for $\mu\in\left( (\Lambda_1+1)^{-1},(\Lambda_0+1)^{-1}\right)$  and   $\widetilde{\eps}>0$, there exists $\delta>0$ such that $|\mu-\widetilde{\mu}|<\delta$ implies,
for every $u\in\ck$, 
$$\left|\llangle A_\mu u, u \rrangle_\ck- \llangle A_{\widetilde{\mu}} u, u \rrangle_\ck\right| \leq \widetilde{\eps} \norm{u}^2_\ck.$$ 
Thus
$$\frac{\llangle A_{\widetilde{\mu}} u, u \rrangle_\ck}{\norm{u}^2_\ck}  \leq  \frac{\llangle A_{\mu} u, u \rrangle_\ck}{\norm{u}^2_\ck}+\widetilde{\eps},$$
and therefore $\kappa(\widetilde \mu) \leq \kappa(\mu)+\widetilde{\eps}$  by \eqref{defkappa}.
Similarly, we obtain the reverse inequality. Together these imply continuity of $\mu\mapsto \kappa(\mu)$.

We next consider monotonicity. Let $u\in\ck$, $$(\Lambda_1+1)^{-1}<\widetilde \mu<\mu< (\Lambda_0+1)^{-1}$$ and $N:= \la (A_\mu-A_{\widetilde \mu}) u, u \ra_\ck.$ Then using Lemma \ref{lemma2},
	\bea
	N&=&  \la P\left[(I-\mu \go^{-1})^{-1}-(I-\widetilde\mu \go^{-1})^{-1}\right] K u , u \ra_\ck \\
	&=& \la \left[(I-\mu \go^{-1})^{-1}-(I-\widetilde\mu \go^{-1})^{-1}\right] K u , K u \ra_{H^{-1}}. 
	\eea
	Let $(v_n)$ be a sequence in $L^2(\Omega)$ such that $v_n\to Ku$ in $H^{-1}_{qp}(\Omega)$. Then 
	\bea
	N &=& \lim_{n\to\infty}  \la \left[(I-\mu \go^{-1})^{-1}-(I-\widetilde\mu \go^{-1})^{-1}\right] v_n , v_n \ra_{H^{-1}}\\
	&=& \lim_{n\to\infty}  \la G_0 \left[(I-\mu \go^{-1})^{-1}-(I-\widetilde\mu \go^{-1})^{-1}\right] v_n , v_n \ra_{L^2}.
	\eea
	As $v_n\in L^2$ we can use the representation of $G_0$ and the resolvents in terms of the Bloch functions, so from \eqref{repf} and \eqref{Blochrep}, we have
	\bea
	N&=&  \lim_{n\to\infty} \int_{-\pi}^\pi \sum_s\frac1{\lambda_s(k)+1}   \left[\frac1{1-\mu(\lambda_s(k)+1)}-\frac{1}{1-\widetilde\mu(\lambda_s({k})+1)}\right]     \left|  \llangle Uv_n(k),\psi_s(k)\rrangle_{L^2}\right|^2 dk.
	\eea
	Since
	\bea &
	\frac{1}{\lambda_s({k})+1}\left[ \frac{1}{1-\mu(\lambda_s({k})+1)}-\frac{1}{1-\widetilde\mu(\lambda_s({ k})+1)}\right] = 
	 \frac{\mu-\widetilde\mu}{(1-\mu(\lambda_s({ k})+1))(1-\widetilde\mu(\lambda_s({ k})+1))}\ >\ 0,
	\eea
	we have $N>0$. 
   Thus, by \eqref{defkappa}, the map  $\mu\mapsto \kappa(\mu)$ is monotonically increasing.
\end{proof}

We  next seek both lower and upper bounds on (\ref{defkappa}).
\subsection{Lower bound.}
\begin{lemma} \label{lem8} For all $u\in\ck$  and $\mu\in\left((\Lambda_1+1)^{-1} , (\Lambda_0+1)^{-1}\right)$ we have
$$\llangle A_\mu u,u\rrangle_{\ck}\geq \dfrac1{1-\mu (\Lambda_1+1)} \norm{Ku}^2_{H^{-1}}.$$
\end{lemma}
\begin{proof}
Let $(v_n)\in (L^2(\Omega))^\N$ such that $v_n \to Ku \in H^{-1}_{qp}(\Omega)$.
Then, as in the proof of Lemma \ref{lem7},
$$\llangle A_\mu u,u \rrangle_{\ck}  = \llangle (I-\mu \go^{-1})^{-1} Ku,Ku \rrangle_{H^{-1}}.$$
 {Using the expansions in terms of Bloch functions \eqref{repf} and \eqref{Blochrep}, we have
\bea 
   \llangle A_\mu u,u \rrangle_{\ck}  &=& \llangle (I-\mu \go^{-1})^{-1} Ku,Ku \rrangle_{H^{-1}}\ = \ \lim_{n\to\infty} \llangle (I-\mu \go^{-1})^{-1} v_n, v_n\rrangle_{H^{-1}}\\
    &=&\lim_{n\to\infty} \llangle  \phi^{-1}  (I-\mu \go^{-1})^{-1} v_n, \phi^{-1} v_n\rrangle_{H^{1}}\ =\ \lim_{n\to\infty} \llangle  G_0  (I-\mu \go^{-1})^{-1} v_n,   v_n\rrangle_{L^2}\\
&=&\lim_{n\to\infty}  \sum_s \int_{-\pi}^\pi \dfrac1{1-\mu (\lambda_s(k)+1)}  \cdot \dfrac1{\lambda_s(k)+1} \left| \llangle U v_n,\psi_s \rrangle \right|^2 dk.
\eea
Next, let $M$ be the index introduced in \eqref{M}. Then, as all terms in the series with $s<M$ are non-negative, we have
\bea
\llangle A_\mu u,u \rrangle_{\ck} &\geq&  \lim_{n\to\infty}\sum_{s\geq M}\int_{-\pi}^\pi   \dfrac1{1-\mu (\lambda_s(k)+1)}  \cdot \dfrac1{\lambda_s(k)+1} \left|\llangle  U v_n,\psi_s \rrangle \right|^2 dk\\
&\geq& \dfrac1{1-\mu (\Lambda_1+1)} \lim_{n\to\infty} \sum_{s\geq M} \int_{-\pi}^\pi \dfrac1{\lambda_s(k)+1} \left| \llangle  Uv_n, \psi_s\rrangle \right|^2 dk.
 \eea
Since $\dfrac1{1-\mu (\Lambda_1+1)}<0$, we can now add the missing bands back in to get
\bea
\llangle A_\mu u,u \rrangle_{\ck} &\geq&    \dfrac1{1-\mu (\Lambda_1+1)} \lim_{n\to\infty} \sum_{s}   \int_{-\pi}^\pi \dfrac1{\lambda_s(k)+1} \left| \llangle Uv_n, \psi_s\rrangle \right|^2 dk \\ 
&=& \lim_{n\to\infty}   \dfrac1{1-\mu (\Lambda_1+1)} \llangle \phi^{-1} v_n, v_n \rrangle_{L^2} \\ 
& =&  \lim_{n\to\infty}   \dfrac1{1-\mu (\Lambda_1+1)} \norm {v_n}^2_{H^{-1}}\ =\ \dfrac1{1-\mu (\Lambda_1+1)} \norm{Ku}^2_{H^{-1}},
\eea
as required.}
\end{proof}

From this, Lemma \ref{lem12} (ii) and (iii) easily lead to 
\beq\label{Rq}
 \dfrac{  \llangle  A_\mu u,u  \rrangle_\ck}{\norm{u}^2_{\ck}}\geq \frac{\norm{\gp}_{H^{-1}\to H^1} \norm{\frac{1}{\eps_0}-\frac{1}{\eps_1}}_\infty}{1-\mu (\Lambda_1+1)}.
\eeq

   \begin{cor} \label{cor1} 
 {If $\norm{\frac{1}{\eps_0}-\frac{1}{\eps_1}}_\infty$ satisfies the estimate \eqref{cond1}, then there exists  $\mu\in\left(\frac{1}{\Lambda_1+1},\frac{1}{\Lambda_0+1}\right)$ such that
 $\dfrac{  \llangle  A_\mu u,u  \rrangle_\ck}{\norm{u}^2_{\ck}}\geq c>-1$ for  all $u\in \ck$.}
 \end{cor}
 
 \begin{proof}
 For $\mu\to(\Lambda_0+1)^{-1}$, the right hand side of \eqref{Rq} tends to a limit, which is greater than $-1$ by \eqref{cond1}.
 \end{proof}
 
Inequality \eqref{Rq} also  shows that for a \textit{fixed} $\mu$ in the spectral gap, the size of the perturbation has to reach a threshold before it is possible for $\mu$ to lie in the spectrum of $\gp$.

\subsection{Upper bound}
Now we show that the minimum of the Rayleigh quotient of $A_\mu$ diverges to $-\infty$
as $\mu$ approaches $\frac{1}{\Lambda_1+1}$ from above. To do this, we have to construct a suitable
test function; in order to bring the interaction with the gap edge into play,
we use the edge Bloch wave $\psi_M$. Here, $M$ is as introduced in \eqref{M}. We recall our assumption that there exists a ball $D$ such that $\eps_1-\eps_0>0$ on $D$. 
\begin{lemma}\label{nonzero}
$(\lo-\lp)\psi_M(\cdot,k_y^0)\neq 0$.
\end{lemma}

\begin{proof}
Assume $(\lo-\lp)\psi_M(\cdot,k_y^0) =  0$. Then 
$$[(\lo-\lp)\psi_M(\cdot,k_y^0)][\psi_M(\cdot,k_y^0)]=  0, \hbox{ so }
\int_\Omega \left(\frac{1}{\eps_0}-\frac{1}{\eps_1}\right) |\nabla \psi_M(\cdot,k_y^0)|^2 =0$$
 and $\nabla \psi_M(\cdot,k_y^0) =0 $ on $D$. Hence $L_0 \psi_M(\cdot,k_y^0) =0$ on $D$. Together with $(L_0-\Lambda_1) \psi_M(\cdot,k_y^0)= 0$ on $\Omega$, this gives $\psi_M(\cdot,k_y^0)=0$ on $D$ and by unique continuation $\psi_M(\cdot,k_y^0)\equiv 0$ (see \cite{Ale12}).
\end{proof}

\begin{remark}
The condition we  require for our results is $(\lo-\lp)\psi_M(\cdot,k_y^0)\neq 0$. We make the assumption  on $\eps_1-\eps_0$ instead, as this can be checked from the data.
\end{remark}

\begin{lemma}\label{testfn}
There exists $u\in\ck$ such that $[(\lo-\lp)\psi_M(\cdot,k_y^0)][\gp u]\neq 0$.
\end{lemma}

\begin{proof}
As $\gp :H^{-1}_{qp}(\Omega)\to H^1_{qp}(\Omega)$ is surjective, by Lemma \ref{nonzero} there exists $\tilde{u}\in H^{-1}_{qp}(\Omega)$ such that $[(\lo-\lp)\psi_M(\cdot,k_y^0)][\gp \tilde{u}]\neq 0$.
Set $u=P\tilde{u}$, then
\bea
[(\lo-\lp)\psi_M(\cdot,k_y^0)][\gp u] &=& \int_\Omega\left(\frac{1}{\eps_0}-\frac{1}{\eps_1}\right) \nabla\psi_M(\cdot,k_y^0)\overline{\nabla \gp u}\\
&=& \overline{[(\lo-\lp)\gp u][\psi_M(\cdot,k_y^0)]} \ =\ \overline{[Ku] [\psi_M(\cdot,k_y^0)]} \\
&=& \overline{[KP\tilde{u}] [\psi_M(\cdot,k_y^0)]} \ =\ \overline{[K\tilde{u}] [\psi_M(\cdot,k_y^0)]}.
\eea 
Reversing all the steps with $u$ replaced by $\tilde{u}$, we get
\bea
[(\lo-\lp)\psi_M(\cdot,k_y^0)][\gp u] &=& [(\lo-\lp)\psi_M(\cdot,k_y^0)][\gp \tilde{u}] \neq 0,
\eea
which completes the proof.
\end{proof}

From now on,  $u$ will always denote the test function in $\ck$ given in Lemma \ref{testfn}. In considering the Rayleigh quotient for our test function, expressions involving $Ku=(\lo-\lp)\gp u$ will arise. To be able to make use of the resolvent representation via Bloch waves in $L^2(\Omega)$, we need to regularize $Ku$. 
First, define
\be
 Tu=\left(\frac{1}{\eps_0}-\frac{1}{\eps_1}\right) \nabla \gp u
\ee
and extend $Tu$ quasi-periodically from $\Omega$ to $\R^2$.
Next, we  introduce a mollifier $(\chi_n)_{n\geq 0}$ with support in $[0,1]^2$ and set
\be
 u_n=Tu *\chi_n\quad \hbox{ and }\quad v_n= -\divg (Tu *\chi_n).
\ee
Then $u_n\vert_{\Omega}, v_n\vert_{\Omega}$ are supported on $\Omega_n$, a neighbourhood of $[0,1]\times[-R,R]$ in $\Omega$ and with $U$ denoting the Floquet-Bloch transform in the $\hat{\by}$-direction, we have for sufficiently large $n$ that
\ben Uu_n(x,y,k) &=& \frac{1}{2\pi}\sum_{m\in\Z}e^{ikm} u_n(x,y-m) \ =\ \frac{1}{2\pi}\sum_{|m|\leq R+1}e^{ikm} u_n(x,y-m). \label{Uun}
\een
Similarly,
\ben \label{Uvn} Uv_n(x,y,k)\ =\ \frac{1}{2\pi}\sum_{|m|\leq R+1}e^{ikm} u_n(x,y-m). 
\een
In particular, in both cases, the sum is finite.

 {We now show that this gives us the desired smooth approximation of $Ku$.}
\begin{lemma}\label{lem:vnKu}
$v_n\to Ku$ in $H^{-1}_{qp}(\Omega)$.
\end{lemma}

\begin{proof}
Let $\varphi\in H^1_{qp}(\Omega)$. Then
$$v_n[\varphi]=\int_{\Omega_n} v_n\overline\varphi = \int_{\Omega_n}\left(\left(\frac{1}{\eps_0}-\frac{1}{\eps_1}\right)\nabla \gp u*\chi_n\right)\overline{\nabla\varphi} =\int_{\Omega_n}\left(Tu*\chi_n\right)\overline{\nabla\varphi}  ,$$
where the boundary term in the integration by parts vanishes, as all functions satisfy quasiperiodic boundary conditions in the $\hat{\bf x}$-direction. On the other hand,
$$Ku[\varphi]=\int_{(0,1)\times(-R,R)} \left(\frac{1}{\eps_0}-\frac{1}{\eps_1}\right)\nabla \gp u \overline{\nabla\varphi}=\int_{(0,1)\times(-R,R)} Tu \overline{\nabla\varphi}.$$
Hence,  {
\bea
\left|(Ku-v_n)[\varphi]\right|&=& \left|\int_{\Omega}\left(Tu-Tu*\chi_n\right) \overline{\nabla\varphi}\right|\ \leq \
\norm{ Tu-Tu*\chi_n}_{L^2} \norm{\varphi}_{H^1}.
\eea
%
%
As $Tu-Tu*\chi_n\to 0$ in $L^2(\Omega)$,  we see that 
$v_n\to Ku$ in $H^{-1}_{qp}(\Omega)$.}
\end{proof}


 {Before finally considering the Rayleigh quotient, we need two more auxilliary results.}

\begin{lemma}\label{lem:unif}
$Uu_n(\cdot,k)\to U(Tu)(\cdot,k)$ uniformly in $k$ in $L^2((0,1)^2)$ as $n\to\infty$.
\end{lemma}

\begin{proof} We consider the expression for $Uu_n$ from \eqref{Uun} and note that
$$ UTu (x,y,k)\ =\ \frac{1}{2\pi}\sum_{|m|\leq R}e^{ikm} Tu(x,y-m).$$
Clearly, we have that $(Tu*\chi_n)(\cdot,\cdot-m)\vert_{(0,1)^2}\to Tu(\cdot,\cdot-m)\vert_{(0,1)^2}$ in $L^2((0,1)^2)$ for $|m|\leq R$ and $\norm{e^{\mp ik}u_n(x,y\pm (R+1))}_{L^2} = \norm{u_n(x,y\pm (R+1))}_{L^2} \to 0$ uniformly in $k$.
\end{proof}

\begin{lemma}\label{lem:pos}
There exist $c>0, \delta>0$ and $N\in\N$ such that
$$\left|\la Uv_n(\cdot,k),\psi_M(\cdot,k)\ra\right|^2\geq c$$
for all $|k-k_y^0|<\delta$ and $n>N$.
\end{lemma}

\begin{proof}
Integrating by parts, we have 
\bea
\la Uv_n(\cdot,k),\psi_M(\cdot,k)\ra &=& \int_{(0,1)^2}  Uu_n(\cdot,k) \cdot \overline{\nabla \psi_M(\cdot,k)} 
\ \to\ \int_{(0,1)^2}  UTu(\cdot,k) \cdot \overline{\nabla \psi_M(\cdot,k)},
\eea
where, by Lemma \ref{lem:unif} the convergence is uniform in $k$.
Now, in view of the location of the   support of the functions, and using the quasi-periodicity of $\psi_M$,
\bea
\int_{(0,1)^2}  UTu(\cdot,k) \cdot \overline{\nabla \psi_M(\cdot,k)} &=& \frac{1}{2\pi} \sum_{|m|\leq R} \int_{(0,1)^2}  e^{ikm} Tu(x,y-m)\overline{\nabla \psi_M(x,y,k)}\\
&=& \frac{1}{2\pi}  \int_{(0,1)\times(-R,R)}  Tu(x,z)\overline{\nabla \psi_M(x,z,k)}
\\
&=& \frac{1}{2\pi} \int_{\Omega} \left(\frac{1}{\eps_0}-\frac{1}{\eps_1}\right) \nabla \gp u \cdot \overline{\nabla \psi_M(\cdot,k)}\\
&=& \frac{1}{2\pi}\overline{\left[ (\lo-\lp) \psi_M(\cdot,k)\right] [\gp u]}.
\eea
By  {Lemma \ref{testfn}}, this is non-zero at $k=k_y^0$.
Consider the map $k\mapsto \left[ (\lo-\lp) \psi_M(\cdot,k)\right] [\gp u]$. This is continuous in $k$ and so there exists $\delta>0$ such that it is non-zero for all $|k-k_y^0|<\delta$. Uniformity of the convergence then proves the result for all $n>N$ for some $N\in\N$.
\end{proof}

\begin{lemma} \label{lem14}
 For the test function $u$ given in Lemma \ref{testfn} we have
 $ \llangle A_\mu u,u \rrangle \to -\infty$ as $\mu \to \dfrac1{\Lambda_1+1}$.
 \end{lemma}
 
 \begin{proof}
 As  in the proof of Lemma \ref{lem8}, and using Lemma \ref{lem:vnKu},
 \begin{align*}
 \llangle A_\mu u,u \rrangle ={}&\lim_{n\to\infty} \sum_s \int_{-\pi}^\pi \dfrac1{1-\mu (\lambda_s(k)+1)} \dfrac1{\lambda_s(k)+1}  \left| \llangle U v_n, \psi_s\rrangle_{L^2} \right|^2  dk \\
{}& \leq \lim_{n\to\infty} \sum_{s\leq M}   \int_{-\pi}^\pi \dfrac1{1-\mu (\lambda_s(k)+1)} \dfrac1{\lambda_s(k)+1}  \left| \llangle U v_n, \psi_s\rrangle_{L^2} \right|^2 dk
\end{align*}

Now, for $\mu$ near $(\Lambda_1+1)^{-1}$,
\begin{align*} 
& \sum_{s<M}    \int_{-\pi}^\pi \dfrac1{1-\mu (\lambda_s(k)+1)} \dfrac1{\lambda_s(k)+1}  \left| \llangle U v_n, \psi_s\rrangle_{L^2} \right|^2 dk \\
  &\quad \leq {} \dfrac1{ 1-\mu (\Lambda_0+1)}  \sum_{s<M} \int_{-\pi}^\pi \dfrac1{\lambda_s(k)+1} \left| \llangle U v_n, \psi_s\rrangle_{L^2} \right|^2 dk\\
   &\quad \leq {}   \dfrac1{ 1-\mu (\Lambda_0+1)}  \sum_s \int_{-\pi}^\pi    \dfrac1{\lambda_s(k)+1}  \left| \llangle U v_n, \psi_s\rrangle_{L^2} \right|^2 dk\\
&\quad  
 {
\leq C \norm{v_n}^2_{H^{-1}} \to C \norm{Ku}^2_{H^{-1}} \leq C\norm{u}^2_{\ck},}
\end{align*}
where we have  {again used Lemma \ref{lem:vnKu} and the last estimate follows by Lemma \ref{lem12} (iii)}.
We are left with the contribution from the $M$-band which we divide up into
integration over two disjoint regions: Let $\delta$ be as in Lemma \ref{lem:pos} and $B_\delta(k_y^0)$ denote the ball of radius $\delta$ around $k_y^0$. Then
$$
\int_{[-\pi,\pi ]\setminus B_\delta(k_y^0)}   \dfrac1{1-\mu (\lambda_M(k)+1)} \dfrac1{\lambda_M(k)+1} \left| \llangle U v_n, \psi_M\rrangle_{L^2} \right|^2 \ dk \leq 0
$$
and using that  {for $|k-k_y^0|<\delta$ we have $\dfrac1{\lambda_M(k)+1} \geq c_1>0$ (by choosing a smaller $\delta$, if necessary)} and $\left| \llangle U v_n, \psi_M\rrangle_{L^2} \right|^2\geq c$  by  Lemma \ref{lem:pos}, we have 
\begin{equation*}
 \int_{B_\delta(k_y^0)}  \dfrac1{ 1-\mu (\lambda_M(k)+1)}  \dfrac1{\lambda_M(k)+1}     \left| \llangle U v_n, \psi_M\rrangle_{L^2} \right|^2\ dk \leq C_\delta \int_{B_\delta(k_y^0)}   \dfrac{dk}{1-\mu (\lambda_M(k)+1)}  \label{eq*}.
\end{equation*}
Now observe that  {from \eqref{condBand} we have}
\bea &&
 \int_{B_\delta(k_y^0)}   \dfrac{dk}{1-\mu (\lambda_M(k)+1)} \ = \ -\int_{B_\delta(k_y^0)}   \dfrac{dk}{\mu (\lambda_M(k)+1)-1}\\ 
&&\leq  { - 
\int_{B_\delta(k_y^0)}   \dfrac{dk}{\mu (\Lambda_1+1 +\alpha|k-k_y^0|^2)-1}}\\
&&\leq -  (\Lambda_1+1)^{-1}
\int_{B_\delta(k_y^0)}   \dfrac{dk}{\mu - 1/(\Lambda_1+1) +\alpha|k-k_y^0|^2/((\Lambda_1+1)(\Lambda_0+1))}.
\eea
The last integral has a nonnegative integrand and has the form 
\beq
\int_{B_\delta(k_y^0)}   \dfrac{dk}{\eta +c_1  {|k-k_y^0|}^2}\label{eq**}
\eeq
with $c_1$ a positive constant and $\eta = \mu - 1/(\Lambda_1+1)  \to 0$ as $\mu \to \dfrac1{\Lambda_1+1}$.
The expression \eqref{eq**} is larger than 
\bea
\int_{ {|k-k_y^0|}\leq \delta_1}  \dfrac{dk}{\eta +c_1  {|k-k_y^0|}^2}\geq \frac{2\delta_1}{\eta+c_1\delta_1^2}
\eea
for any $0<\delta_1\leq \delta$. By setting $\delta_1^2=\eta$, we see that the integral diverges as $\eta\to 0$.
Thus finally, $\la A_\mu u, u \ra\to -\infty$ as $\mu\to 1/(\Lambda_1+1)$.
\end{proof}

Combining the results of Lemma \ref{lem:cont}, Lemma \ref{lem8} and Lemma \ref{lem14}, we obtain our main result, Theorem \ref{theorem}, from the Intermediate Value Theorem. In particular, any arbitrarily weak perturbation induces spectrum into the gap.

%
\section{Concluding Remarks and open problems}
%

We provided a sufficient rigorous criterion for localization in gaps by arbitrarily weak
line defects, for the case of TE-polarized electromagnetic waves. We arrive at our results
by comparing the Green's operators of the perturbed and unperturbed systems. While
Green's functions techniques have been a part of the theoretical physics literature
for a long time (see e.g. \cite{Eco}), our method combines Green's functions and variational
methods. For example, we do not use series expansion of the difference of the
operators $G_0$ and $G_1$ to get approximations, and the variational approach avoids
in an elegant way the need to control the remainder terms.

The method presented here is, in principle, also applicable to both the case when the band edge under consideration is degenerate and to the full Maxwell equations, at the expense of greater technical complexity. We plan to deal with these in forthcoming work. 

Another open problem is the following. We know now sufficient conditions to create gap modes which are localized in the $\hat \by$-direction centering on the line defect. If the modes were additionally localized in the $\hat \bx$-direction, we would have a bound state of the operator $-\nabla\cdot \eps_1^{-1}\nabla$ on the whole of $\R^2$. This would go against physical intuition, since then light would stand still in the defect. 
It would be desirable to show that there are no modes that are localized in the $\hat \bx$-direction, i.e.~the perturbation creates truly guided modes. This would equivalently mean, that there is no flat band created in the gap (for a discussion, see \cite{KuchWav1}).
The absence of bound states for periodic Helmholtz operators with line defects has been proven in \cite{HoRa}. However, to show absence of bound states for periodic divergence type operators seems to be extremely difficult. For periodic operators with sufficiently smooth coefficients, this question is investigated addressed in \cite{FK15}.

%
\section{Compact embedding of $H^1_{qp}(\Omega, e^{\gamma|y|})$ in $L^2(\Omega)$}
%

In this appendix, we briefly sketch the compact embedding of $H^1_{qp}(\Omega, e^{\gamma|y|})$ in $L^2(\Omega)$  {for $\gamma>0$}.
For any $f\in H^1_{qp}(\Omega, e^{\gamma|y|})$,
\beq\label{eq:app}
\int_{ {\Omega, |y|\geq R}} |f|^2 \leq e^{-\gamma R} \int_{ {\Omega, |y|\geq R}}e^{\gamma|y|} |f|^2 \leq  e^{-\gamma R}  \|f\|_{H^1_{qp}(\Omega, e^{\gamma|y|})}^2.
\eeq
Let $f_j = f^{(1)}_j$ be a bounded sequence in $H^1_{qp}(\Omega, e^{\gamma|y|})$. Let $\Omega_p:= (0, 1)\times(-p, p)$
for any $p\in \N$.
Since $H^1_{qp}(\Omega_p)$ embeds compactly into $L^2(\Omega_p)$, we may extract from $(f^{(1)}_j)$ a subsequence
$(f^{(2)}_j)$ converging in $L^2(\Omega_1)$ and from $(f^{(2)}_j)$ a subsequence converging in $L^2(\Omega_2)$ and
so forth. We claim that the diagonal sequence $(f^{(p)}_p)$ is Cauchy in $L^2(\Omega)$. This is seen as follows:
given any $\eps > 0$, determine first a $p_0$ so large that 
$$
\left(\int_{ {\Omega, |y|\geq p_0}} |f^{(p)}_p|^2 dy\right)^{1/2} \leq \frac{\eps}{3}.
$$
for all $p\geq p_0$, using \eqref{eq:app}. Now determine a $p_1\geq p_0$ so large that $\|f^{(p_0)}_{p}-f^{(p_0)}_{q}\|_{L^2(\Omega_{p_0})}\leq \eps/3$
for all $p, q\geq p_1$.
Since $(f^{(p)}_p)_{p\geq p_0}$ is a subsequence of $(f^{(p_0)}_j)$, we have for $p, q\geq p_1$
$$
\|f^{(p)}_p - f^{(q)}_q\|_{L^2(\Omega)}\leq \| f^{(p)}_p - f^{(q)}_q \|_{L^2(\Omega_{p_0})}+\|f^{(p)}_p\|_{L^2(\Omega\setminus \Omega_{p_0})}+\|f^{(q)}_q\|_{L^2(\Omega\setminus \Omega_{p_0})} \leq \eps.
$$

\section*{Acknowledgements}
The authors would like to thank 
Prof.~Carlos Kenig for pointing out the paper \cite{Ale12} on unique continuation.
VH expresses his gratitude to
the German Research Foundation (DFG) for continued support through grants
FOR 5156/1-1 and FOR 5156/1-2. He also acknowledges partial support by 
NSF grant NSF-DMS 1412023.

\end{document}